\documentclass{new_tlp}
\usepackage{mathptmx}

\usepackage{times}
\usepackage{xcolor}
\usepackage{soul}
\usepackage[T1]{fontenc}
\usepackage[utf8]{inputenc}
\usepackage[small]{caption}
\usepackage{amsmath}
\usepackage{amssymb}
\usepackage{multicol}
\usepackage{url}
\usepackage{tikz}
\usepackage{pgfplots}
\usepackage{subcaption}
\usepackage{todonotes}

\tikzset{
    dot/.style 2 args={fill, circle}
}

\clearpage{}%

\newenvironment{changemargin}[2]{%
\list{}{\rightmargin#2\leftmargin#1
\parsep=0pt\topsep=0pt\partopsep=0pt}
\item[]}
{\endlist}

\newenvironment{indented}{\begin{changemargin}{1cm}{0cm}}{\end{changemargin}}

\newtheorem{theorem}{Theorem}

\newtheorem{proposition}[theorem]{Proposition}

\newtheorem{definition}[theorem]{Definition}
\newtheorem{example}[theorem]{Example}

\let\phi\varphi
\let\epsilon\varepsilon

\newcommand{\qedsymbol}{\hfill\ensuremath{\square}}

\newcommand{\calA}{\mathcal{A}}

\newcommand{\calC}{\mathcal{C}}

\newcommand{\calM}{\mathcal{M}}

\newcommand{\calR}{\mathcal{R}}

\newcommand{\calT}{\mathcal{T}}

\newcommand{\NP}{\ensuremath{\textsc{NP}}}

\newcommand{\SIGMA}[2]{\ensuremath{\Sigma_{\mathit{#1}}^{\mathit{#2}}}}

\newcommand{\constant}[1]{\mathit{#1}}

\newcommand{\variable}[1]{\mathit{#1}}
\newcommand{\variables}[1]{{\mathbf{#1}}}
\newcommand{\term}[1]{\mathit{#1}}
\newcommand{\terms}[1]{\mathbf{#1}}
\newcommand{\mods}[1]{\mathit{models}(#1)}
\newcommand{\answersets}[1]{\mathit{AS}(#1)}

\newcommand{\termt}{\term{t}}
\newcommand{\termst}{\terms{t}}
\newcommand{\varV}{\variable{V}}
\newcommand{\varW}{\variable{W}}
\newcommand{\varX}{\variable{X}}
\newcommand{\varY}{\variable{Y}}
\newcommand{\varZ}{\variable{Z}}

\newcommand{\varsX}{\variables{X}}
\newcommand{\varsY}{\variables{Y}}
\newcommand{\varsZ}{\variables{Z}}

\newcommand{\relation}[1]{{\mathit{#1}}}

\newcommand{\fullatom}[2]{{\relation{#1}(#2)}}

\newcommand{\varof}[1]{\mathit{var}(#1)}
\newcommand{\elitof}[1]{\mathit{elit}(#1)}

\newcommand{\eneg}{\mathbf{not}\,}

\newcommand{\body}[1]{{\mathit{B}(#1)}}
\newcommand{\nbody}[1]{{\mathit{B}^-(#1)}}
\newcommand{\pbody}[1]{{\mathit{B}^+(#1)}}
\newcommand{\head}[1]{{\mathit{H}(#1)}}

\clearpage{}%

\newcommand{\linkselp}{\url{https://dbai.tuwien.ac.at/proj/selp}}
\newcommand{\linkbench}{\url{https://dbai.tuwien.ac.at/proj/selp}}

\title[Theory and Practice of Logic Programming]
      {\emph{selp}: A Single-Shot Epistemic Logic Program Solver}

\author[M.\ Bichler, M.\ Morak, and S.\ Woltran]
       {Manuel Bichler, 
	Michael Morak, and
	Stefan Woltran\\
	TU Wien, Vienna, Austria\\
	\{surname\}@dbai.tuwien.ac.at}

\begin{document}

\maketitle

\begin{abstract}
Epistemic Logic Programs (ELPs) are an extension of Answer Set Programming
(ASP) with epistemic operators that allow for a form of meta-reasoning, that
is, reasoning over multiple possible worlds. Existing ELP solving approaches
generally rely on making multiple calls to an ASP solver in order to evaluate
the ELP. However, in this paper, we show that there also exists a direct
translation from ELPs into non-ground ASP with bounded arity. The resulting
ASP program can thus be solved in a single shot. We then implement this
encoding method, using recently proposed techniques to handle large,
non-ground ASP rules, into the prototype ELP solving system ``selp'', which we
present in this paper. This solver exhibits competitive performance on a set
of ELP benchmark instances.
 \end{abstract}

\section{Introduction}\label{sec:introduction}

Epistemic Logic Programs (ELPs), as defined in \cite{ai:ShenE16}, are an
extension of the well-established formalism of Answer Set Programming (ASP). ASP
is a generic, fully declarative logic programming language that allows us to
encode problems in such a way that the resulting answers (called \emph{answer
sets}) directly correspond to solutions of the encoded problem
\cite{cacm:BrewkaET11}.  Negation in ASP is generally interpreted according to
the stable model semantics \cite{iclp:GelfondL88}, that is, as
\emph{negation-as-failure} or \emph{default negation}. The default negation
$\neg a$ of an atom $a$ is true if there is no justification for $a$ in the same
answer set, making it a ``local'' operator in the sense that it is defined
relative to the same answer set. ELPs, on the other hand, extend ASP with the
\emph{epistemic negation} operator $\eneg$ that allows for a form of
meta-reasoning, that is, reasoning over multiple answer sets.  Intuitively, an
epistemically negated atom $\eneg a$ expresses that $a$ cannot be \emph{proven}
true, that is, it is false in at least one answer set.  Thus, epistemic negation
is defined relative to a collection of answer sets, referred to as a \emph{world
view}. The main reasoning task for ELPs, checking that a world view exists, is
\SIGMA{P}{3}-complete~\cite{ai:ShenE16}.

Epistemic negation has long been recognized as a desired construct for ASP
\cite{aaai:Gelfond91,amai:Gelfond94}. In these works, Michael Gelfond introduced
the modal operators $\mathbf{K}$ (``known'' or ``provably true'') and
$\mathbf{M}$ (``possible'' or ``not provably false''), in order to address this
need. Given an atom $a$, $\mathbf{K}a$ and $\mathbf{M}a$ stand for $\neg \eneg
a$ and $\eneg \neg a$, respectively.

\begin{example}\label{ex:introduction}
  A classical example for the use of epistemic negation is the presumption of
  innocence rule
  $$\fullatom{innocent}{\varX} \gets \eneg \fullatom{guilty}{\varX},$$
  namely: a person is innocent if they cannot be proven guilty.
\end{example}

Renewed interest in recent years has revealed
several flaws in the original semantics, and several approaches (cf.\ e.g.\
\cite{lpnmr:Gelfond11,birthday:Truszczynski11,diss:Kahl14,ijcai:CerroHS15,ai:ShenE16})
aimed to refine them in such a way that unintended world views are
eliminated. In this work, we will settle on the semantics proposed in
\cite{ai:ShenE16}.
The flurry of new research also led to the development of ELP solving systems
\cite{logcom:KahlWBGZ15,ijcai:SonLKL17%
}. %
Such 
solvers %
employ readily available, highly efficient ASP systems like
\emph{clingo} \cite{book:GebserKKS12,iclp:GebserKKS14} and \emph{WASP}
\cite{lpnmr:AlvianoDFLR13}, especially making use of the former solver's
multi-shot solving functionality \cite{corr:GebserKKS17}.
However, these ELP solving systems %
rely on
ground ASP programs
when calling the ASP solver, which, for reasons rooted in complexity theory, generally requires 
multiple calls in order to check for world
view existence. The main aim of our paper is to present techniques and a system
for solving ELPs that is able to utilize an ASP solver in such a way that the
ELP can be solved in a single shot.

\paragraph{Contributions.} Our contributions in this paper are twofold:
\smallskip

\begin{itemize}
  \item 
    We propose a novel translation from ELPs to ASP programs using
    large non-ground ASP rules, such that the ELP can be solved by an ASP
    solving system in a single shot. This is done via a recently
    proposed encoding technique \cite{tplp:BichlerMW16} that uses
    large ASP rules to formulate complex checks. This technique builds on a
    result from 
    \cite{amai:EiterFFW07} 
    that states that evaluating
    non-ground ASP programs with bounded predicate arity is
    \SIGMA{P}{3}-complete, which matches the complexity of evaluating ELPs. Our
    proposed translation is therefore optimal from a complexity-theoretic point
    of view. From a practical point of view, such an encoding avoids multiple
    calls to the ASP solver. State-of-the-art systems use sophisticated
    heuristics and learning, and multiple calls %
	might result in a loss
	of 
	knowledge about the problem instance, which the solver has
    already learned.

  \item 
    We further discuss how our encoding needs to be constructed in order to be
    useful in practice. In particular, in current ASP systems, non-ground ASP
    programs first need to be \emph{grounded}, that is, all variables need to be
    replaced by all allowed combinations of constants. Since our encoding makes
    use of large non-ground rules, a naive grounding will often not terminate,
    since there may be hundreds or thousands of variables in a rule. However, as
    proposed in \cite{tplp:BichlerMW16}, we make use of the \emph{lpopt} rule
    decomposition tool \cite{lopstr:BichlerMW16} that splits such large rules 
    into smaller ones that are more easily grounded, by making use of the
    concept of \emph{treewidth} and \emph{tree decompositions} \cite{Bodlaender93}. To
    use this tool to its full potential, %
    the large rules we use in our
    encoding must be constructed carefully, in order for \emph{lpopt} to
    split them up optimally.

  \item 
	Finally, we present a prototype implementation of our ELP-to-ASP
    rewriting approach and combine it with the state-of-the-art ASP solving
    system \emph{clingo} \cite{iclp:GebserKKS14} in order to evaluate its
    performance. We compare our system against \emph{EP-ASP}~\cite{ijcai:SonLKL17}
    on different benchmarks found in the
    literature. Our system shows competitive performance on these benchmarks, in particular
	on instances with good structural properties.
\end{itemize}
\smallskip

The remainder of the paper is structured as follows: in
Section~\ref{sec:preliminaries} we introduce the formal background of ASP, ELPs,
and tree decompositions; Section~\ref{sec:encoding} states our reduction from
ELPs to ASP, including practical considerations and a discussion 
of related work; Section~\ref{sec:qbf} presents
how QBF formulas can be encoded as ELP programs; Section~\ref{sec:system}
introduces our ELP solver; Section~\ref{sec:experiments} presents our benchmark
results, making use of results from Section~\ref{sec:qbf}; and finally
Section~\ref{sec:conclusions} closes with some concluding remarks.

\smallskip This paper is an extended versions of
\cite{BichlerMW18,BichlerMW18a}.  Additional material includes a full
correctness proof for our reduction in Section~\ref{sec:encoding} and a
formalized and detailed description of the adaptations needed to make the
reduction workable in practice. Further, Section~\ref{sec:qbf} describes in
detail how our QBF benchmarks, used in Section~\ref{sec:experiments}, are
constructed.
\section{Preliminaries}\label{sec:preliminaries}

\paragraph{Answer Set Programming (ASP).} A \emph{ground logic program} (also
called answer set program, ASP program, or simply program) is a pair $\Pi =
(\calA,\calR)$, where $\calA$ is a set of propositional (i.e.\ ground) atoms and
$\calR$ is a set of rules of the form
\begin{equation}\label{eq:rule}
  a_1\vee \cdots \vee a_l \leftarrow a_{l+1}, \ldots, a_m, \neg a_{m+1}, \ldots,
  \neg a_n;
\end{equation}
where the comma symbol stands for conjunction, $n \geq m \geq l \geq 0$ and $a_i
\in \calA$ for all $1 \leq i \leq n$. Each rule $r \in \calR$ of
form~(\ref{eq:rule}) consists of a head $\head{r} = \{ a_1,\ldots,a_l \}$ and a
body given by $\pbody{r} = \{a_{l+1},\ldots,a_m \}$ and $\nbody{r} =
\{a_{m+1},\ldots,a_n \}$. A \emph{literal} $\ell$ is either an atom $a$ or its
(default) negation $\neg a$. A literal $\ell$ is true in a set of atoms $I
\subseteq \calA$ if $\ell = a$ and $a \in I$, or $\ell = \neg a$ and $a \not\in
I$; otherwise $\ell$ is false in $I$.  A set $M \subseteq \calA$ is a called a
\emph{model} of $r$ if $\pbody{r} \subseteq M$ together with $\nbody{r} \cap M =
\emptyset$ implies that $\head{r} \cap M \neq \emptyset$. We denote the set of
models of $r$ by $\mods{r}$ and the models of a logic program $\Pi=
(\calA,\calR)$ are given by $\mods{\Pi} = \bigcap_{r \in \calR} \mods{r}$. 
The GL-reduct $\Pi^I$ of a ground logic program $\Pi$ with respect to a set of atoms $I
\subseteq \calA$ is the program $\Pi^I = \left(\calA,\left\{\head{r}
\leftarrow \pbody{r} \mid r \in \calR, \nbody{r} \cap I =
\emptyset\right\}\right)$. %

\begin{definition}\label{def:answerset}
\cite{iclp:GelfondL88,ngc:GelfondL91}
  $M \subseteq \calA$ is an \emph{answer set} of a program $\Pi$ if (1) $M
  \in \mods{\Pi}$ and (2) there is no subset $N \subset M$ such that $N \in
  \mods{\Pi^M}$.
\end{definition}

The set of answer sets of a %
program $\Pi$ is denoted $\answersets{\Pi}$.
The \emph{consistency problem} of ASP (decide whether, %
given
$\Pi$, $\answersets{\Pi}\neq\emptyset$) is
\SIGMA{P}{2}-complete~\cite{amai:EiterG95}. 

General \emph{non-ground logic programs} differ from ground logic programs in
that variables may occur in rules. Such rules are $\forall$-quantified
first-order implications of the form $H_1\vee\dots\vee H_k\leftarrow
P_1,\dots,P_n,\neg N_1,\dots,\neg N_m$ where $H_i$, $P_i$ and $N_i$ are
(non-ground) atoms. A \emph{non-ground atom} $A$ is of the form
$\fullatom{p}{\termst}$ and consists of a predicate name $\relation{p}$, as well
as a sequence of \emph{terms} $\termst$, where each term $\termt \in \termst$ is
either a \emph{variable} or a \emph{constant} from a domain $\Delta$, with
$|\termst|$ being the arity of $\relation{p}$. Let $\varof{A}$ denote the set of
variables $\varsX$ in a non-ground atom $A$. This notation naturally extends to
sets. We will denote variables by capital letters, constants and predicates by
lower-case words. A non-ground rule can be seen as an abbreviation for all
possible instantiations of the variables with domain elements from $\Delta$.
This step is usually explicitly performed by a \emph{grounder} that transforms a
(non-ground) %
logic program into a set of ground rules of the
form~(\ref{eq:rule}). Note that, in general, such a ground %
program may be
exponential in the size of the non-ground %
program. For non-ground programs of bounded arity, the consistency problem is
\SIGMA{P}{3}-complete \cite{amai:EiterFFW07}.

\paragraph*{Epistemic Logic Programs.} A \emph{ground epistemic logic program
(ELP)} is a pair $\Pi = (\calA, \calR)$, where $\calA$ is a set or propositional
atoms and $\calR$ is a set of rules of the following form:
\begin{equation*}
   a_1\vee \cdots \vee a_k \leftarrow \ell_1, \ldots, \ell_m, \xi_1, \ldots,
  \xi_j, \neg \xi_{j + 1}, \ldots, \neg \xi_{n},
\end{equation*}
where each $a_i$ is an atom, each $\ell_i$ is a \emph{literal}, and each $\xi_i$
is an \emph{epistemic literal}, that is, a formula $\mathbf{not}\,\ell$, where
$\mathbf{not}$ is the epistemic negation operator, and $\ell$ is a literal.
W.l.o.g.\ we assume that no atom appears twice in a rule\footnote{This can be
achieved by introducing auxiliary atoms whenever an atom appears twice in a
rule, and add two rules that ensure that the original and auxiliary atom must be
equivalent.}. Let $\elitof{r}$ denote the set of all epistemic literals
occurring in a rule $r \in \calR$. This notation naturally extends to programs.
Let $\head{r} = \{ a_1, \ldots, a_k \}$. Let $\body{r} = \{ \ell_1, \ldots,
\ell_m, \xi_1, \ldots, \xi_j, \neg \xi_{j+1}, \ldots, \neg \xi_{n} \}$, that is,
the set of elements appearing in the rule body.

In order to define the main reasoning tasks for ELPs, we recall the notion of
the epistemic reduct \cite{ai:ShenE16}. Let $\Phi \subseteq \elitof{\Pi}$
(called a \emph{guess}). The \emph{epistemic reduct} $\Pi^\Phi$ of the program
$\Pi = (\calA, \calR)$ w.r.t.\ $\Phi$ consists of the rules $\{ r^\neg \mid r
\in \calR \}$, where $r^\neg$ is defined as the rule $r$ with all epistemic
literals $\eneg \ell$ in $\Phi$ (resp.\ in $\elitof{\Pi} \setminus \Phi$)
replaced by $\top$ (resp.\ $\neg \ell$). Note that $\Pi^\Phi$ is a logic program
without epistemic negation\footnote{We interpret double negation according to
\cite{ai:FaberPL11}.}. This leads to the following, central definition.

\begin{definition}\label{def:worldview}
  Let $\Phi$ be a guess. The set $\calM = \answersets{\Pi^\Phi}$ is called a
  \emph{candidate world view} of $\Pi$ iff
  \begin{enumerate}
	\item\label{def:worldview:1} $\calM \neq \emptyset$,
	  
	\item\label{def:worldview:2} for each epistemic literal
	  $\mathbf{not}\,\ell \in \Phi$, there exists an answer set $M \in
	  \calM$ wherein $\ell$ is false, and 

	\item\label{def:worldview:3} for each epistemic literal
	  $\mathbf{not}\,\ell \in \elitof{\Pi} \setminus \Phi$, it holds that
	  $\ell$ is true in each answer set $M \in \calM$.
  \end{enumerate}
\end{definition}

\begin{example}\label{ex:running1}
  Let $\Pi$ be the following ELP, with $\calR = \{ r_1, r_2 \}$:
  \begin{align*}
	r_1: p \gets \eneg q\\
	r_2: q \gets \eneg p
  \end{align*}
  ELP $\Pi$ has two candidate world views: (1) $\Phi = \{ \eneg q \}$
  with $\answersets{\Pi^\Phi} \,{=}\, \{ \{ p \} \}$; %
(2) $\Phi = \{ \eneg p \}$
  with $\answersets{\Pi^\Phi}\, {=}\, \{ \{ q \} \}$.\qedsymbol
\end{example}

The main reasoning task we treat in this paper is the \emph{world view existence
problem} (or \emph{ELP consistency}), that is, given an ELP $\Pi$, decide
whether a candidate world view exists. This problem is known to be
\SIGMA{3}{P}-complete \cite{ai:ShenE16}.

\paragraph*{Tree Decompositions.} A \emph{tree decomposition} of a graph $G =
(V,E)$ is a pair $\calT = (T, \chi)$, where $T$ is a rooted tree and $\chi$ is a
labelling function over nodes $t$, with $\chi(t) \subseteq V$, such that the
following holds: (i) for each $v \in V$ there is a node $t$ in $T$ such that $v
\in \chi(t)$; (ii) for each $\{v,w\} \in E$ there is a node $t$ in $T$ such that
$\{v, w\} \subseteq \chi(t)$; and (iii) for all nodes $r$, $s$, and $t$ in $T$,
where $s$ lies on the path from $r$ to $t$, $\chi(r) \cap \chi(t) \subseteq
\chi(s)$.  The \emph{width} of a tree decomposition $\calT$ is defined as the
maximum cardinality of $\chi(t)$ minus one, over all nodes $t$ of $\calT$. The
\emph{treewidth} of a graph $G$ is the minimum width over all tree
decompositions of~$G$. Trees have treewidth 1, cliques of size $k$ have
treewidth $k$. Finding a tree decomposition of minimal width is \NP-hard in
general.
\section{Single-Shot ELP Solving}\label{sec:encoding}

In this section, we provide our novel translation for solving ELPs via a single
call to an ASP solving system. The goal is to transform a given ELP $\Pi$ to a
non-ground ASP program $\Pi'$ with predicates of bounded arity, such that $\Pi$
is consistent (i.e.\ it has a candidate world view) iff $\Pi'$ has at least one
answer set. A standard ASP solver can then decide the consistency problem for
the ELP $\Pi$ in a single call, by solving $\Pi'$.

\subsection{Reducing ELPs to ASP Programs}\label{sec:reduction}

The reduction is based on an encoding technique proposed in
\cite{tplp:BichlerMW16}, which uses large, non-ground rules. Given an ELP $\Pi$,
the ASP program $\Pi'$ will roughly be constructed as follows.  $\Pi'$ contains
a guess part that chooses a set of epistemic literals from $\elitof{\Pi}$,
representing a guess $\Phi$ for $\Pi$.  Then, the check part verifies that, for
$\Phi$, a candidate world exists. In all, the ASP program $\Pi'$ consists of
five parts: $$\Pi' = \Pi'_\mathit{facts} \cup \Pi'_\mathit{guess} \cup
\Pi'_{\mathit{check}_{\ref{def:worldview:1}}} \cup
\Pi'_{\mathit{check}_{\ref{def:worldview:2}}} \cup
\Pi'_{\mathit{check}_{\ref{def:worldview:3}}},$$ where the sub-program
$\Pi'_\mathit{facts}$ is a set of facts representing the ELP $\Pi$, and
$\Pi'_{\mathit{check}_i}$ represents the part of the program that checks
Condition~$i$ of Definition~\ref{def:worldview}. We now proceed to the
construction of the program $\Pi'$. Let $\Pi = (\calA, \calR)$ be the ELP to
reduce from. To ease notation, let $\calA = \{ a_1, \ldots, a_n \}$.

\paragraph{The set of facts $\Pi'_\mathit{facts}$.} %
$\Pi'_\mathit{facts}$ %
represents basic knowledge
about the ELP $\Pi$, plus some auxiliary facts, and is given as:
\begin{itemize}
  \item $\fullatom{atom}{a}$, for each atom $a \in \calA$;

  \item $\fullatom{elit}{\ell}$, for each epistemic literal $\mathbf{not}\,\ell
    \in \elitof{\Pi}$\footnote{Note that we use the literal $\ell$ as an ASP
    constant.};

  \item $\fullatom{leq}{0, 0}$, $\fullatom{leq}{0, 1}$, and $\fullatom{leq}{1, 1}$,
    representing the less or equal relation for Boolean values; and

  \item $\fullatom{or}{0, 0, 0}$, $\fullatom{or}{0, 1, 1}$, $\fullatom{or}{1, 0,
    1}$, and $\fullatom{or}{1, 1, 1}$, representing the Boolean relation \emph{or}.
\end{itemize}

\paragraph{Sub-Program $\Pi'_\mathit{guess}$.} This part of the program
consists of a single, non-ground rule that guesses a subset of the epistemic
literals (stored in predicate $\relation{g}$) as follows:
$$\fullatom{g}{\variable{L}, 1} \vee \fullatom{g}{\variable{L}, 0} \gets
\fullatom{elit}{\variable{L}}.$$

\paragraph{Shorthands.} Before defining the three check parts of the program, we
will introduce some useful shorthands which will be used at several occasions. 
To this end, we use a context identifier $\calC$. We first define the following:
$$H_\mathit{val}^\calC(\variable{A}) \equiv \fullatom{v_\calC}{\variable{A}, 1}
\vee \fullatom{v_\calC}{\variable{A}, 0},$$ that is,
$H_\mathit{val}^\calC(\variable{A})$ guesses a truth assignment for some
variable $\variable{A}$ and stores it in relation $\relation{v_\calC}$. We will
often use variables $\varsX = \{ \varX_1, \ldots, \varX_n \}$ or $\varsY = \{
\varY_1, \ldots, \varY_n \}$ to represent a subset $M$ of $\calA$, where
assigning $\varX_i$  to $1$ characterizes $a_i \in M$, and $\varX_i = 0$
otherwise. Let $$B_\mathit{val}^\calC(\varsX) \equiv \bigwedge_{a_i \in \calA}
\fullatom{v_\calC}{a_i, \varX_i},$$ that is, $B_\mathit{val}^\calC(\varsX)$
extracts the truth assignment from relation $\relation{v_\calC}$ into the
variables $\varsX$ as described above. Finally, for some rule $r$ in $\Pi$, we
define a formula
$B_\mathit{sat}^r(\varsX, \varsY, \variable{S})$ that checks whether
the rule $r$ is satisfied in the epistemic reduct $\Pi^\Phi$ w.r.t.\ the guess
$\Phi$ encoded in the relation $\relation{g}$, when the negative body (resp.\ 
positive body and head) is evaluated over the set of atoms encoded by $\varsX$ (resp.\
$\varsY$). If the rule is satisfied, $B_\mathit{sat}^r(\varsX, \varsY, 1)$
should hold, and $B_\mathit{sat}^r(\varsX, \varsY, 0)$ otherwise.  This is done
as follows. Let $r$ contain the atoms $\{ a_{i_1}, \ldots,a_{i_m} \}$
(recall that no atom appears twice in a rule), where $i_1, \ldots, i_m \in \{ 1,
\ldots, n \}$. For ease of notation, we will use a four-ary $\relation{or}$
relation, which can easily be split into two of our three-ary $\relation{or}$
atoms using a helper variable $\variable{T}$:
$$\fullatom{or}{\varW, \varX, \varY, \varZ} %
\gets
\fullatom{or}{\varW, \varX,
\variable{T}}, \fullatom{or}{\variable{T}, \varY, \varZ}.$$ The following is the
central building block of our reduction:
 
\newcommand{\smin}{\scalebox{0.65}[1.0]{\( - \)}}
\begin{multline*}
  B_\mathit{sat}^r(\varsX, \varsY, \variable{R}_m) \equiv \variable{R}_0
  = 0, \bigwedge_{a_{i_j} \in \head{r}} \fullatom{or}{\variable{R}_{j-1},
  \varY_{i_j}, \variable{R}_j},\\
  \bigwedge_{a_{i_j} \in \body{r}} \fullatom{or}{\variable{R}_{j-1},
  1\smin\varY_{i_j}, \variable{R}_j}, \bigwedge_{\neg a_{i_j} \in \body{r}}
  \fullatom{or}{\variable{R}_{j-1}, \varX_{i_j}, \variable{R}_j},\\
  \bigwedge_{\eneg a_{i_j} \in \body{r}} \fullatom{g}{a_{i_j}, \variable{N}_j},
  \fullatom{or}{\variable{N}_j, 1\smin\varX_{i_j}, \variable{T}_j},
  \fullatom{or}{\variable{R}_{j-1}, 1\smin\variable{T}_j, \variable{R}_j},\\
  \bigwedge_{\eneg \neg a_{i_j} \in \body{r}} \fullatom{g}{\neg a_{i_j},
  \variable{N}_j}, \fullatom{or}{\variable{N}_j, \varY_{i_j}, \variable{T}_j},
  \fullatom{or}{\variable{R}_{j-1}, 1\smin\variable{T}_j, \variable{R}_j},\\
  \bigwedge_{\neg \eneg a_{i_j} \in \body{r}} \fullatom{g}{a_{i_j},
  \variable{N}_j}, \fullatom{or}{\variable{R}_{j-1}, \variable{N}_j,
  1\smin\varY_{i_j}, \variable{R}_j},\\
  \bigwedge_{\neg \eneg \neg a_{i_j} \in \body{r}} \fullatom{g}{\neg a_{i_j},
  \variable{N}_j}, \fullatom{or}{\variable{R}_{j-1}, \variable{N}_j,
  \varX_{i_j}, \variable{R}_j}.
\end{multline*}

For a rule $r$, each big conjunction in the above formula encodes a reason for
$r$ to be satisfied. For example, the fifth line encodes the fact that rule $r$
is true if the disjunct $\neg \eneg a_{i_j}$ is not satisfied, that is, if the
epistemic literal $\eneg a_{i_j}$ is part of the guess $\Phi$, or the atom
$a_{i_j}$ is false (represented by $1\smin\varY_{i_j}$). Each disjunct of rule
$r$ is evaluated in this way, and the results are connected via the
$\relation{or}$ relation (with the result of the first $i$ disjuncts stored in
variable $\variable{R}_i$). Therefore, $\variable{R}_m$ will be $1$ if $r$ is
satisfied, and $0$ otherwise, as desired (recall that $r$ has $m$ disjuncts).
The following example illustrates how this shorthand is constructed for a
concrete input program.

\begin{example}\label{ex:running2}
  Recall program $\Pi = (\calA, \calR)$ from Example~\ref{ex:running1}. Let
  $\calA = \{ a_1, a_2 \}$, where $a_1 = p$ and $a_2 = q$.  Let rule $r_2 \in
  \calR$ contain the atoms $\{ a_{i_1}, a_{i_2} \}$, where $i_1 = 2$ and $i_2 =
  1$. We give the core construct, $B_\mathit{sat}^{r}(\cdot, \cdot,
  \cdot)$ for rule $r_2$:
  \begin{multline*}
    B_\mathit{sat}^{r_2}(\varX_1, \varX_2, \varY_1, \varY_2, R_2) \equiv
    \fullatom{or}{0, \varY_2, \variable{R}_1}, \fullatom{g}{\constant{p},
    \variable{N}_2}, \fullatom{or}{\variable{N}_2, 1\smin\varX_1,
    \variable{T}_2}, \fullatom{or}{\variable{R}_1, 1\smin\variable{T}_2,
    \variable{R}_2}.
  \end{multline*}\qedsymbol
\end{example}

Finally, we define $B_\mathit{ss}(\varsX, \varsY)$, which
makes sure that the variables $\varsY$ identify a strict subset of the atoms
identified by $\varsX$. Let $B_\mathit{ss}(\varsX, \varsY) \equiv$ $$
\variable{N}_0 = 0, \variable{N}_n = 1, \bigwedge_{a_i \in \calA}
\fullatom{leq}{\varY_i, \varX_i}, \fullatom{or}{\variable{N}_{i-1}, \varX_i
\smin \varY_i, \variable{N}_i}.$$ We can now proceed with the remainder of our
reduction.

\paragraph{Sub-Program $\Pi'_{\mathit{check}_{\ref{def:worldview:1}}}$.} This part
of the program needs to check that, given the guess $\Phi$ made in
$\Pi'_\mathit{guess}$, there exists at least one answer set of the epistemic
reduct $\Pi^\Phi$, as per Definition~\ref{def:worldview}(\ref{def:worldview:1}).
Therefore, according to Definition~\ref{def:answerset}, we need to find a set $M
\subseteq \calA$, such that (1) $M$ is a model of $\Pi^\Phi$, and (2) there is
no proper subset of $M$ that is a model of the GL-reduct $(\Pi^\Phi)^M$.
$\Pi'_{\mathit{check}_{\ref{def:worldview:1}}}$ contains the following rules:

\begin{itemize}
  \item $H_\mathit{val}^{\constant{check}_{\ref{def:worldview:1}}}(\variable{A})
	\gets \fullatom{atom}{\variable{A}}$;
	
      \item $\bot \gets B_\mathit{val}^{\constant{check}_{\ref{def:worldview:1}}}(
	\varsX), B_\mathit{sat}^r(\varsX, \varsX, 0)$, for each %
	$r \in \calR$;
	and

      \item $\bot \gets B_\mathit{red}^{\constant{check}_{\ref{def:worldview:1}}}$.
\end{itemize}
The first rule guesses a truth assignment for all atoms. The second rule
verifies that there is no rule in $\Pi^\Phi$ that is violated by the candidate
answer set $M$, represented by the variables $\varsX$, guessed by the first
rule. $B_\mathit{red}^\calC$ checks whether a subset of $M$ is a model of the
GL-reduct $(\Pi^\Phi)^M$. To this end, let $$B_\mathit{red}^\calC \equiv
B_\mathit{val}^\calC(\varsX), B_\mathit{ss}(\varsX, \varsY), \bigwedge_{r \in
\calR} B_\mathit{sat}^r(\varsX, \varsY, 1).$$ 

The last big conjunction in $B_\mathit{red}^\calC$ makes sure that the subset $N
\subset M$ identified by the variables $\varsY$ is indeed a model of every rule
in the GL-reduct $(\Pi^\Phi)^M$. This completes
$\Pi'_{\mathit{check}_{\ref{def:worldview:1}}}$.

\paragraph{Sub-Program $\Pi'_{\mathit{check}_{\ref{def:worldview:2}}}$.} This part
needs to check that, for every epistemic literal $\eneg \ell \in \Phi$, the
epistemic reduct $\Pi^\Phi$ has some answer set wherein $\ell$ is
false. $\Pi'_{\mathit{check}_{\ref{def:worldview:2}}}$ contains the following
rules and facts, for each epistemic literal $\eneg \ell \in \elitof{\Pi}$ (used
as the context $\calC$ so guesses are independent):
\begin{itemize}
  \item $H_\mathit{val}^\ell(\variable{A}) \gets \fullatom{atom}{\variable{A}},
    \fullatom{g}{\ell, 1}$;
  
  \item $\fullatom{v_\ell}{a, \eta}$, where $\eta = 1$ if $\ell =
	\neg a$, or $\eta = 0$ if $\ell = a$;

  \item $\bot \gets B_\mathit{val}^\ell(\varsX), B_\mathit{sat}^r(\varsX,
	\varsX, 0)$, for each $r \in \calR$; and

  \item $\bot \gets B_\mathit{red}^\ell$.

\end{itemize}

These rules guess, for each epistemic literal $\eneg \ell \in \Phi$, a candidate
answer set $M$ wherein $\ell$ is false, and then verify that $M$ is indeed an
answer set, using the same technique as in
$\Pi'_{\mathit{check}_{\ref{def:worldview:1}}}$. This ensures
Condition~\ref{def:worldview:2} of Definition~\ref{def:worldview}.

\paragraph{Sub-Program $\Pi'_{\mathit{check}_{\ref{def:worldview:3}}}$.} Finally,
this part needs to check that, for every epistemic literal $\eneg \ell \in
\elitof{\Pi} \setminus \Phi$, every answer set of $\Pi^\Phi$ satisfies $\ell$.
The construction makes use of the technique of
saturation \cite{amai:EiterG95}:
\begin{itemize}
  \item $H_\mathit{val}^{\mathit{check}_{\ref{def:worldview:3}}}(\variable{A})
    \gets \fullatom{atom}{\variable{A}}$;

  \item $\fullatom{v_{\mathit{check}_{\ref{def:worldview:3}}}}{\variable{A}, 0} \gets \mathit{sat},
	\fullatom{atom}{\variable{A}}$;

      \item $\fullatom{v_{\mathit{check}_{\ref{def:worldview:3}}}}{\variable{A}, 1}
	\gets \mathit{sat}, \fullatom{atom}{\variable{A}}$; and

  \item $\bot \gets \neg\mathit{sat}$.
\end{itemize}

This setup checks that, for every candidate answer set $M$ guessed in the first
rule, the atom $\mathit{sat}$ is derived. Since we are only interested in answer
sets, we first check that $M$ is indeed one, using the following rules,
similarly to $\Pi'_{\mathit{check}_{\ref{def:worldview:1}}}$:
\begin{itemize}
  \item $\mathit{sat} \gets
    B_\mathit{val}^{\mathit{check}_{\ref{def:worldview:3}}}(\varsX),
	B_\mathit{sat}^r(\varsX, \varsX, 0)$, for each %
	$r \in \calR$; and

  \item $\mathit{sat} \gets
    B_\mathit{red}^{\mathit{check}_{\ref{def:worldview:3}}}$.
\end{itemize}

It now remains to check that in each answer set $M$ (that is, where
$\mathit{sat}$ has not been derived yet) all epistemic literals $\eneg \ell$ are
either in $\Phi$, or otherwise $\ell$ is true in $M$. This is done by adding the
following rule to $\Pi'_{\mathit{check}_{\ref{def:worldview:3}}}$:
 
\begin{multline*}
  \mathit{sat} \gets \bigwedge_{\eneg a \in \elitof{\Pi}} \fullatom{g}{a,
    \variable{N}_a}, \fullatom{v_{\mathit{check}_{\ref{def:worldview:3}}}}{a,
  \varX_a}, \fullatom{or}{\variable{N}_a, \varX_a, 1},\\
  \bigwedge_{\eneg \neg a \in \elitof{\Pi}} \fullatom{g}{\neg a,
    \variable{N}^\neg_a}, \fullatom{v_{\mathit{check}_{\ref{def:worldview:3}}}}{a,
      \varX^\neg_a},
  \fullatom{or}{\variable{N}^\neg_a, 1\smin\varX^\neg_a, 1}.
\end{multline*}

This completes the reduction. We will now show that this reduction indeed
accomplishes our goals. The correctness of our reduction can be intuitively seen
from the observation that each of the three check parts of the constructed ASP
program $\Pi'$ ensures precisely one of the three conditions that define a
candidate world view. Each answer set $A$ of $\Pi'$ is a witness for the fact
that a guess $\Phi \subseteq \elitof{\Pi}$ encoded in $A$ indeed gives rise to a
candidate world view. The next theorem formally states that our reduction is
correct.

\begin{theorem}
  Let $\Pi = (\calA, \calR)$ be an ELP and let $\Pi'$ be the ASP program
  obtained from $\Pi$ via the above reduction. Then, $\Pi$ has a candidate world
  view if and only if $\Pi'$ has an answer set.
\end{theorem}

\begin{proof}
  We will begin with the ``if'' direction. To this end, assume that there is a
  guess $\Phi \subseteq \elitof{\Pi}$ for $\Pi$ that gives rise to a candidate
  world view $\calM = \answersets{\Pi^\Phi}$. We will show that $\Pi'$ has an
  answer set $M$. Clearly, $M$ contains all the facts from
  $\Pi'_\mathit{facts}$. Furthermore, let $M$ contain the fact
  $\fullatom{g}{\ell, 1}$ for each epistemic literal $\eneg \ell \in \Phi$ and the fact
  $\fullatom{g}{\ell, 0}$ for each epistemic literal $\eneg \ell \in \elitof{\Pi}
  \setminus \Phi$. This clearly satisfies sub-program $\Pi'_\mathit{guess}$.

  Now, let $M' \in \calM$ be any answer set of $\Pi^\Phi$ (such an answer set
  exists, since, by assumption, $\calM$ is a candidate world view for $\Phi$ and
  by Definition~\ref{def:worldview}, $\calM$ is non-empty). Let $M$ contain the
  fact $\fullatom{v_{\mathit{check}_{\ref{def:worldview:1}}}}{a, 1}$ for each $a
  \in M'$ and the fact $\fullatom{v_{\mathit{check}_{\ref{def:worldview:1}}}}{a,
  0}$ for each $a \in \calA \setminus M'$. This satisfies sub-program
  $\Pi'_{\mathit{check}_{\ref{def:worldview:1}}}$ as follows. Clearly, $M$
  satisfies the first line of the sub-program. Since the atoms with relation
  $\relation{v_{\mathit{check}_{\ref{def:worldview:1}}}}$ encode precisely the
  answer set $M'$ of $\Pi^\Phi$, and since $M'$ is a model of $\Pi^\Phi$, also
  the second line of the sub-program is satisfied, which, by construction,
  checks that the assignment encoded in relation
  $\relation{v_{\mathit{check}_{\ref{def:worldview:1}}}}$ satisfies all the rules
  of $\Pi^\Phi$. Finally, the third line, by construction, checks that the same
  assignment is also minimal w.r.t.\ the GL-reduct $[\Pi^\Phi]^{M'}$. Since $M'$
  is an answer set of $\Pi^\Phi$, also this line of the sub-program
  $\Pi'_{\mathit{check}_{\ref{def:worldview:1}}}$ is satisfied.

  The argument for satisfaction of $\Pi'_{\mathit{check}_{\ref{def:worldview:2}}}$
  is similar to the one for $\Pi'_{\mathit{check}_{\ref{def:worldview:1}}}$. Since
  $\calM$ is a candidate world view for guess $\Phi$, it contains, for each
  epistemic literal $\eneg \ell \in \Phi$, an answer set $M_\ell \in \calM$ such
  that $\ell$ is false in $M_\ell$. Thus, the argument for sub-program
  $\Pi'_{\mathit{check}_{\ref{def:worldview:1}}}$ can be analogously applied for
  each $\eneg \ell \in \Phi$, taking the answer set $M_\ell$ instead of $M'$.

  Finally, we need to verify that $M$ also satisfies the rules in
  $\Pi'_{\mathit{check}_{\ref{def:worldview:3}}}$. To this end, let $M$ contain
  the facts $M_{\mathit{check}_{\ref{def:worldview:3}}}$ consisting of the fact
  $\relation{sat}$, as well as the fact
  $\fullatom{v_{\mathit{check}_{\ref{def:worldview:3}}}}{a, b}$ for each $a \in
  \calA$ and $b \in \{ 0, 1 \}$. It is easy to verify that all the rules in
  $\Pi'_{\mathit{check}_{\ref{def:worldview:3}}}$ are classically satisfied.
  However, since the negative literal $\neg \relation{sat}$ appears in line 4,
  in order to verify that $M$ is indeed an answer set, we also need to look at
  minimality w.r.t.\ GL-reduct. Since line 4 is removed in the GL-reduct of
  $\Pi'_{\mathit{check}_{\ref{def:worldview:3}}}$, it may be the case that some
  subset of $M_{\mathit{check}_{\ref{def:worldview:3}}} \setminus \{
  \relation{sat} \}$ may indeed satisfy the GL-reduct. However, we will show
  that every such subset requires $\relation{sat}$ to be true via lines 5, 6, or
  7 of $\Pi'_{\mathit{check}_{\ref{def:worldview:3}}}$ (i.e., those rules
  of $\Pi'_{\mathit{check}_{\ref{def:worldview:3}}}$ with atom $\relation{sat}$ in the head),
  and can therefore not
  exist. Indeed, every subset of $M_{\mathit{check}_{\ref{def:worldview:3}}}
  \setminus \{\relation{sat}\}$ that does not encode an answer set of $\Pi^\Phi$ in
  relation $\relation{v_{\mathit{check}_{\ref{def:worldview:3}}}}$ derives
  $\relation{sat}$ via lines 5 or 6, by construction (the argument to see this
  is analogous to the one for the previous two sub-programs). It remains to show
  that all other remaining subsets (i.e.\ the subsets of
  $M_{\mathit{check}_{\ref{def:worldview:3}}} \setminus \{\relation{sat}\}$ that
  encode answer sets of $\Pi^\Phi$) also derive $\relation{sat}$. However, since
  every answer set $M' \in \calM$, by Definition~\ref{def:worldview}, has to
  satisfy precisely the condition encoded by line 7 of
  $\Pi'_{\mathit{check}_{\ref{def:worldview:3}}}$, this is easy to see. We thus
  have that $M$, as constructed above, is indeed an answer set of $\Pi'$.

  The ``only if'' direction can be seen via similar arguments to the above. By
  construction, any answer set $M$ of $\Pi'$ will encode a guess $\Phi$ for
  $\Pi$. Since any such answer set $M$, to be an answer set, must satisfy the
  three check sub-programs of $\Pi'$ in the way described above, and these three
  check sub-programs, by construction, correspond directly to the three
  conditions of Definition~\ref{def:worldview}, we have that $M$ encodes a guess
  $\Phi$ for $\Pi$ that leads to a candidate world view.
\end{proof}

As we have seen, our reduction works as intended: the ASP program $\Pi'$ derived
from the input ELP $\Pi$ has an answer set precisely when $\Pi$ has a candidate
world view. The next interesting observation is that our reduction is, in fact,
a polynomial-time reduction, as stated below.

\begin{theorem}
  Given an ELP $\Pi$, the reduction above runs in time $O(e \cdot n)$, where $n$
  is the size of $\Pi$ and $e = |\elitof{\Pi}|$, and uses predicates of arity at
  most three.
\end{theorem}

\begin{proof}
  Predicates of arity at most three are used if the four-ary $\relation{or}$
  relation is not materialized as an actual relation in ASP, but viewed as a
  shorthand for two connected ternary $\relation{or}$ relations (cf.\ the
  paragraph on shorthands of our reduction).
  The reduction's runtime (and output size) can be seen to be in $O(e \cdot n)$
  by noting the fact that the construct $B_\mathit{red}$ is of size linear in
  $n$ (it precisely encodes each rule using the $\relation{or}$ predicates).
  $B_\mathit{red}$ is then used once for each epistemic literal in $\Pi$ (cf.\
  $\Pi'_{\mathit{check}_{\ref{def:worldview:2}}}$).
\end{proof}

Note that the above theorem shows that our reduction is indeed worst-case
optimal as claimed in Section~\ref{sec:introduction}: checking consistency of
non-ground, fixed-arity ASP programs is \SIGMA{P}{3}-complete, as is checking
world view existence for ELPs.

\subsection{Using the Reduction in Practice}\label{sec:reduction-practice}

As we have seen, using the construction in the previous subsection, we can solve
the consistency problem for a given ELP via a single call to an ASP solving
system. However, when trying this in practice, the performance is less than
optimal, mainly for the following reason. At several points in the construction,
large non-ground rules are used (i.e.\ where $B_\mathit{red}^\calC$ appears in a
rule body). As noted in Section~\ref{sec:preliminaries}, these rules need to be
grounded, but may contain hundreds or thousands of variables, which need to be
replaced by all possible combinations of constants; a hopeless task for ASP
grounders, as the resulting ground program is exponential in the number
of variables.

However, as noted in \cite{tplp:BichlerMW16}, such large rules can often be
decomposed into smaller, more manageable rules, using the \emph{lpopt} tool
\cite{lopstr:BichlerMW16}. This tool roughly works as follows: (1) compute a
\emph{rule graph} $G_r$ for each non-ground rule $r$, where there is a vertex
for each variable $\varV$ in $r$, and there is an edge between $\varV_1$ and
$\varV_2$, if the two variables appear together in an atom of $r$; then (2)
compute a tree decomposition of $G_r$ of minimal width; and finally, (3) in a
bottom-up manner, output a rule for each node in the tree decomposition. The
resulting rules each contain only as many variables as the treewidth of $G_r$
(plus one), and, together, are equivalent to the original rule $r$. After this
rule decomposition step, grounding now becomes much easier, since the number of
variables in each rule is reduced. Note that, since finding optimal tree
decompositions is NP-hard, \emph{lpopt} employs heuristics to find good
decompositions.

In our construction, $B_\mathit{red}^\calC$ stands for a long rule body that
effectively encodes the entire input ELP $\Pi$. Each atom $a_i$ in $\Pi$ is
represented by the two variables $X_i$ and $Y_i$. If we represent $\Pi$ as a
graph $G_\Pi$, where each atom $a_i$ is a vertex, and there is an edge between
two atoms if they appear together in a rule in $\Pi$, then this graph structure
can be found (as a minor) in the rule graph of $B_\mathit{red}^\calC$. However,
in addition, $B_\mathit{red}^\calC$ also adds a series of $\fullatom{or}{\cdot,
\cdot, \cdot}$ atoms (via $B_\mathit{ss}(\varsX, \varsY)$), that introduce
additional connections in the rule graph of $B_\mathit{red}^\calC$. These
connections may increase the treewidth substantially. In fact, even if $G_\Pi$
has a treewidth of 1, by introducing the additional connections in a bad way,
the treewidth may increase arbitrarily: imagine that $G_\Pi$ is a chain,
depicted in black in Figure~\ref{fig:chaingrid}, and imagine the
$\fullatom{or}{\cdot, \cdot, \cdot}$-chain from $B_\mathit{ss}(\varsX, \varsY)$
is inserted into $G_\Pi$, illustrated in pink. The treewidth now depends on the
chain's length (and thereby on the size of $\Pi$), and \emph{lpopt} can no
longer split the rule well.

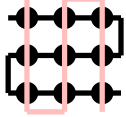
\begin{figure}[t]
  \centering
  \begin{tikzpicture}[scale=0.5]

    \def\maxX{2}
    \def\secondcolor{pink}
    \foreach \x [count=\n] in {0,...,\maxX}{
        \draw[line width=2pt] (-0.5,\x) -- (\maxX+0.5,\x);
        \foreach \y in {0,...,\maxX}{
            \node[dot] at (\x,\y) {};
        }
    }
    \draw[line width=2pt] (\maxX+0.5,\maxX) -- (\maxX+0.5,\maxX-1);
    \draw[line width=2pt] (-0.5,\maxX-1) -- (-0.5,\maxX-2);
    \foreach \x [count=\n] in {0,...,\maxX}{
        \draw[line width=2pt, color=\secondcolor] (\x,-0.5) -- (\x,\maxX+0.5);
    }
    \draw[line width=2pt, color=\secondcolor] (\maxX,\maxX+0.5) -- (\maxX-1,\maxX+0.5);
    \draw[line width=2pt, color=\secondcolor] (\maxX-1,-0.5) -- (\maxX-2,-0.5);

  \end{tikzpicture}
\vspace{-0.10cm}
  \caption{Creating a grid from chains.}
  \label{fig:chaingrid}
\vspace{-0.40cm}
\end{figure}

In the following, we will formalize the problem described above and present an
extension to our reduction presented in the previous subsection that will
alleviate the problem. First, we define the primal graph of an ELP $\Pi$, a
standard notion in topics of satisfiability, constraint programming, and logic
programming; cf.\ standard textbooks, e.g.\ \cite{book:EbbinghausF95}.

\begin{definition}
  The \emph{primal graph} of an ELP $\Pi = (\calA, \calR)$ is the graph $G_\Pi =
  (V, E)$, where $V = \calA$ and there is an edge $(a_i, a_j) \in E$ iff the
  atoms $a_i$ and $a_j$ occur together in a rule in $\calR$.
\end{definition}

Then, we define the rule graph for a non-ground ASP rule $r$:

\begin{definition}
  The \emph{rule graph} of a non-ground ASP rule $r$ is the graph $G_r = (V,
  E)$, such that $V = \varof{r}$, and there is an edge between two variables
  $\varX$ and $\varY$ in $E$ iff $\varX$ and $\varY$ occur together in an atom
  in $r$.
\end{definition}

From the construction, it is not difficult to see that any rule $r$ containing
$B_\mathit{red}^\calC$ reflects the structure of the input ELP $\Pi$, or, more
formally, the graph $G_\Pi$ is contained (as a minor) in the graph $G_r$. Thus,
by well-known graph-theoretic results, the treewidth of $G_r$ is at least the
treewidth of $G_\Pi$. Since this is an integral part of our construction, we
cannot hope for $\emph{lpopt}$ to split up rule $r$ any better than the
structure of $\Pi$ allows. However, as noted in the intuitive problem
description above, $G_r$ contains additional connections between variables.
These are introduced by the subformula $B_\mathit{ss}(\cdot, \cdot)$ that
effectively links all the variables in a rule $r$ into a chain in $G_r$. In the
worst case, as illustrated by Figure~\ref{fig:chaingrid}, these additional
connections in $G_r$ may increase the treewidth arbitrarily, making it almost
impossible for \emph{lpopt} to split up the rule well. It is therefore important
to introduce these additional connections carefully. We will now introduce a
more involved construction of $B_\mathit{ss}(\cdot, \cdot)$ that preserves the
treewidth of $G_\Pi$ in $G_r$ (i.e.\ does not arbitrarily increase it). In this
modified version, $B_\mathit{ss}(\cdot, \cdot)$ is constructed as follows:

\begin{enumerate}
  \item First, compute a tree decomposition $\calT_\Pi$ of $G_\Pi$ with minimal
    width.
	
  \item Secondly, construct $B_\mathit{ss}(\cdot, \cdot)$ in a bottom-up (i.e.\
	post-order traversal) fashion along this tree decomposition in the following
	way, for each node type. To this end, let $\calA = \{ a_1, \ldots, a_n \}$,
	and, for a node $t$ of $\calT_\Pi$, let $\chi(t)$ contain the set of atoms
	$\{ a_{i_1}, \ldots, a_{i_m} \}$, with $i_j \in \{ 1, \ldots, n \}$.

	\begin{description}
	  \item[Leaf Node $t$:] For a leaf node $t$ of $\calT_\Pi$, let
		$B_\mathit{ss}(\varsX, \varsY)$ contain the following conjunction of
		atoms: $$N^t_0 = 0, \variable{N}^t = \variable{N}^t_m,
		\bigwedge_{a_{i_j} \in \chi(t)} \fullatom{leq}{\varY_{i_j},
		\varX_{i_j}}, \fullatom{or}{\variable{N}^t_{j-1}, \varX_{i_j} \smin
		\varY_{i_j}, \variable{N}^t_j},$$ that is, $\variable{N}^t$ contains
		$1$, if the proper subset condition between $\varsX$ and $\varsY$ is
		already fulfilled in node $t$, and $0$ otherwise.

	  \item[Inner Node $t$:] For an inner node $t$ of $\calT_\Pi$ with children
		$t_1, \ldots, t_k$, let $B_\mathit{ss}(\varsX, \varsY)$ contain the same
		conjunction as for a leaf node, but where the equality atom
		$\variable{N}^t = \variable{N}^t_m$ is replaced by the following
		disjunction: $$\fullatom{or}{\variable{N}^t_m, \variable{N}^{t_1},
		\ldots, \variable{N}^{t_k}, \variable{N}^t},$$ where the $k+2$-ary
		$\relation{or}$ atom can be split into 3-ary $\relation{or}$
		atoms in the same way as with the 4-ary $\relation{or}$ atom in
		our main construction. Intuitively, we now have that
		$\variable{N}^t$ is set to $1$ if the proper subset condition is
		already fulfilled somewhere in the subtree rooted at $t$.

	  \item[Root Node $t_\mathit{root}$:] Finally, for the root node
		$t_\mathit{root}$, we add the same conjunction of atoms to
		$B_\mathit{ss}(\varsX, \varsY)$ as for an inner node, but, in addition,
		we add the final condition $\variable{N}^{t_\mathit{root}} = 1$, that
		makes sure that, at the root node, the proper subset condition is
		fulfilled.
	\end{description}
\end{enumerate}

If constructed in the way described above, it is not difficult to see that
$B_\mathit{ss}(\varsX, \varsY)$ still ensures the same condition as in our
original construction from Section~\ref{sec:reduction}, namely, that the
variables $\varsY$ identify a proper subset of the atoms identified by the
variables $\varsX$. However, the treewidth of a rule containing
$B_\mathit{ss}(\cdot, \cdot)$ is now not increased arbitrarily. In fact, it can
be verified that the treewidth of $G_\Pi$ is preserved up to a constant additive
factor, for any rule containing $B_\mathit{red}^\calC$, when using the
alternative construction for $B_\mathit{ss}(\cdot, \cdot)$ provided above. In
practice, this means that \emph{lpopt} is able to split the rule up as well as
possible; that is, as well as the structure of $\Pi$ allows.

\subsection{Discussion and Related Work}\label{sec:relatedwork}

As we have seen, the reduction proposed above allows us to solve ELPs via a
single call to an ASP solving system. However, our encoding also has several
other interesting practical properties, which make it very flexible for use
with, for example, different ASP semantics, or harder problems. A brief
discussion follows.

\paragraph{Other ASP Semantics.} Apart from the original
semantics for ASP (called \emph{stable model semantics},
\cite{iclp:GelfondL88,ngc:GelfondL91}), several different semantics have been
proposed that investigate how to interpret more advanced constructs in ASP, like
double negation, aggregates, optimization, etc \cite{amai:LifschitzTT99,%
amai:Pearce06,tplp:PelovDB07,ai:FerrarisLL11,ai:FaberPL11,ai:ShenWEFRKD14}.
Epistemic reducts may contain double negation, and we have opted to use the FLP
semantics by Faber et al.\ \shortcite{ai:FaberPL11}, as used by
Shen and Eiter \shortcite{ai:FaberPL11}, to interpret
this. The actual interpretation of double negation is encoded in the
$B_\mathit{sat}^r(\cdot, \cdot)$ shorthand defined in our reduction. This
construction is very flexible and can easily be modified to use different ASP
semantics (e.g.\ \cite{amai:LifschitzTT99}).

\paragraph{Enumeration of World Views.} Modern ASP systems like
\emph{clasp} \cite{ai:GebserKS12} contain several useful features not included
in the ASP base language. One such feature is an advanced implementation of
projection, as presented in \cite{cpaior:GebserKS09}: given a set of atoms (or
relations), the solver will output answer sets where all other atoms are
projected away, and will also guarantee that there are no repetitions (even if
multiple answer sets with the same assignment on the projected atoms exist),
while still maintaining efficiency.  This can be used to enumerate candidate
world views by projecting away all relations in our encoding, except for
$\relation{g}(\cdot)$ and
$\relation{v_{\mathit{check}_{\ref{def:worldview:1}}}}(\cdot)$. When enumerating
all projected answer sets in this way, our encoding yields all guesses together
with their candidate world views (when grouped by $\relation{g}(\cdot)$).

\paragraph{Comparison to Related Work.} Classic ELP solvers
generally work by first establishing a candidate epistemic guess $\Phi$ and then
use an answer set solver to verify that the epistemic guess indeed yields an
epistemic reduct whose answer sets form a candidate world view of the original
ELP w.r.t.\ $\Phi$. Different approaches are used to find promising epistemic
guesses, and also to verify that they lead to candidate world views, but,
generally, these systems have in common that an underlying ASP solver is used,
and called multiple times, to solve the ELP. Notable recent ELP solvers include
that follow this approach include \emph{EP-ASP} \cite{ijcai:SonLKL17},
\emph{GISolver} \cite{aspocp:ZhangWZ15} and a later, probabilistic, variant
called \emph{PelpSolver}, and \emph{ELPsolve} \cite{aspocp:KahlLS16}. A
comprehensive survey of recent ELP solving systems (including the one presented
in the present paper) can be found in \cite{corr:LeclercK18}.

We are not aware of another single-shot ELP solver that only needs to call an
underlying ASP system once. However, the idea of our approach is similar to the
one used by Bichler et al. \shortcite{tplp:BichlerMW16}, where a single-shot ASP encoding for
disjunctive ASP, which is rewritten into non-ground normal ASP with fixed arity,
is presented. That is, a solving system for normal ASP would be able to solve a
disjunctive ASP program in a single call. However, this approach was not
implemented, and only presented as an example to show how long non-ground rules
with fixed arity can be used to solve hard problems.  In order to use such
encodings (including our own presented herein), \cite{tplp:BichlerMW16} make use
of \emph{rule decomposition}, where large non-ground ASP rules are split up into
smaller parts based on tree decompositions \cite{iclp:MorakW12}. This rule
decomposition approach has been implemented as a stand-alone tool called
\emph{lpopt} \cite{lopstr:BichlerMW16}, but has recently also been integrated
into ASP solving systems like \emph{I-DLV} \cite{ia:CalimeriFPZ17}.
\section{Application: QBF Solving}
\label{sec:qbf}

In this section, we illustrate the power of ELPs by illustrating a way to solve
QBF formulas with at most three quantifier alternations (3-QBF) by encoding them
as ELPs. This provides an alternative way to show the $\SIGMA{P}{3}$\ lower
bound for ELP consistency, but relies on the existence of a reduction from
so-called \emph{restricted 3-QBF formulas}. \cite{ai:ShenE16} present such a
reduction from restricted 3-QBF formulas to ELP world view existence. Our aim is
to generalize this by presenting a reduction from (general) 3-QBF formulas to
restricted 3-QBF formulas. We will use this result to benchmark our ELP solver
presented in Section~\ref{sec:system}. Let us begin by first recalling the
definition of such 3-QBF formulas.

\begin{definition}\label{def:qbf}
  A \emph{3,$\exists$-QBF in CNF form} (or \emph{QBF}, for short) is a formula
  of the form $$\exists \varsX \forall \varsY \exists \varsZ \phi$$ where
  $\varsX$, $\varsY$, and $\varsZ$ are sets (or sequences) of distinct
  (propositional) atoms (also called variables), and $\phi = \bigwedge_{i=1}^k
  C_i$ is a CNF over the atoms $\varsX \cup \varsY \cup \varsZ$, i.e. $C_i =
  \bigvee_{j=0}^{k_i} L_{i,j}$ is a clause of size $k_i$ and $L_{i,j}$ is either
  an atom $a$ or its negation $\neg a$.
\end{definition}

W.l.o.g.\ we can assume that the clause size $k_i = 3$ for each $0 < i \leq k$,
that is, that $\phi$ is given in 3-CNF form, where each clause has at most three
elements. In \cite{ai:ShenE16}, the authors make use of a version of QBFs
called \emph{restricted QBFs}. These are QBFs that evaluate to true under all
interpretations of the existentially quantified variables if all universally
quantified variables are replaced by $\top$ (i.e.\ if they are set to true).

\begin{definition}
  A \emph{restricted QBF} is a QBF where $\phi[y/\top \mid y \in \varsY]$ is a
  tautology.
\end{definition}

The hardness proof of Theorem~5 of \cite{ai:ShenE16} is a reduction from the
validity problem of restricted QBFs to the consistency problem of epistemic
logic programs. While the actual construction of the reduction is not needed for
our purposes in this section, we nevertheless report it here, for completeness
sake.

\begin{proposition}\cite[Proof of Theorem~5]{ai:ShenE16}
\label{prop:sheneiter}
  Let $\Theta=\exists \varsX \forall \varsY \exists \varsZ \phi$ be a restricted QBF.
  Then, there exists an ELP $\Pi$ such that $\Pi$ has a candidate world view iff
  $\Theta$ is satisfiable.
\end{proposition}

\begin{proof}
  The ELP $\Pi$ consists of the following rules:
  \begin{multicols}{2}
  \begin{itemize}
    \item For each variable $\varX \in \varsX$: \[ \varX \gets \eneg
      \overline{\varX}, \] \[ \overline{\varX} \gets \eneg \varX. \]
    \item For each variable $\varY \in \varsY$: \[ \varY \gets \neg
      \overline{\varY}, \] \[ \overline{\varY} \gets \neg \varY. \]
    \item For each variable $\varZ \in \varsZ$: \[ \varZ \vee \overline{\varZ}.
      \]
    \item For each clause $C_i$, $0 < i \leq k$: \[ U \gets L_{i,1}^*,
      L_{i,2}^*, L_{i,3}^*, \] where $^*$ is an operator that converts a
      positive literal $\varW$ into $\overline{\varW}$, and a negative literal
      $\neg \varW$ into $\varW$.
    \item For each $\varZ \in \varsZ$: \[ \varZ \gets U, \] \[\overline{\varZ}
      \gets U. \]
    \item And, finally, the rule \[ V \gets \eneg V, \eneg \neg U. \]
  \end{itemize}
  \end{multicols}
  ELP $\Pi$ has a (candidate) world view iff $\exists \varsX \forall \varsY
  \exists \varsZ \, \phi$ is satisfiable \cite{ai:ShenE16}.
\end{proof}

We now show a more general reduction that also works for the non-restricted
case. To this end, we will combine the \cite{ai:ShenE16} reduction with our own
reduction of QBF formulas to restricted QBF formulas. %
To achieve our
goal, we are going to introduce one new atom $v_i$ in each clause $C_i$ and
$\forall$-quantify these new atoms together with the $\varsY$ atoms.

\begin{definition}
  Given a QBF $\Theta = \exists \varsX \forall \varsY \exists \varsZ \phi$ with
  $\phi$ being constructed as in Definition \ref{def:qbf}, let its
  \emph{extension}, denoted $\Theta^\uparrow$, be the QBF $$\Theta^\uparrow =
  \exists \varsX \forall (\varsY \cup \variables{V}) \exists \varsZ \phi'$$
  where $$\phi' = \bigwedge_{i=1}^k \left( v_i \lor \bigvee_{j=0}^{k_i}{L_{i,j}}
  \right)$$ and $\variables{V} = \{v_1, \ldots, v_k\}$ is a list of fresh atoms.
\end{definition}

It is easy to see that any extension $\Theta^\uparrow$ of a QBF $\Theta$ is a
restricted QBF.

\begin{proposition}
  Let $\Theta$ be a QBF. Its extension $\Theta^\uparrow$ is a restricted QBF.
\end{proposition}

We will now show that validity-equivalence between a QBF and its extension is
preserved. For the proof, we establish the following terminology: given a subset
of atoms $\sigma \subseteq \variables{S}$, we define its \emph{out-set}
$\overline{\sigma} = \{\neg a \mid a \in \variables{S} \setminus \sigma\}$ and
its \emph{literal-set} $\widehat{\sigma} = \sigma \cup \overline{\sigma}$.

\begin{proposition}\label{prop:validityequivalence}
  $\Theta$ and $\Theta^\uparrow$ are validity-equivalent.
\end{proposition}

\begin{proof}
  ($\Rightarrow$) Assume $\Theta$ is valid, i.e.\ there exists an interpretation
  $\sigma_X \subseteq \varsX$, such that for any interpretation $\sigma_Y
  \subseteq \varsY$ there exists an interpretation $\sigma_Z \subseteq \varsZ$
  such that $(\widehat{\sigma_X} \cup \widehat{\sigma_Y} \cup
  \widehat{\sigma_Z}) \cap C_i \neq \emptyset$ for all $i \in \{1, \ldots, k\}$.
  By monotonicity of non-emptiness of set intersections, also
  $(\widehat{\sigma_X} \cup \widehat{\sigma_Y} \cup \widehat{\sigma_V} \cup
  \widehat{\sigma_Z}) \cap (C_i \cup \{v_i\})\neq \emptyset$ for any
  interpretation $\sigma_V \subseteq \variables{V}$ of a list of new atoms
  $\variables{V} = \{ v_1, \ldots, v_k \}$. But this is proof of the validity of
  $\Theta^\uparrow$.

  ($\Leftarrow$) For the other direction, assume $\Theta^\uparrow$ is valid, i.e.\
  there exists an interpretation $\sigma_X \subseteq \varsX$ such that for any
  interpretations $\sigma_Y \subseteq \varsY$ and $\sigma_V \subseteq
  \variables{V}$ there exists an interpretation $\sigma_Z \subseteq \varsZ$ such
  that $(\widehat{\sigma_X} \cup \widehat{\sigma_Y} \cup \widehat{\sigma_V} \cup
  \widehat{\sigma_Z}) \cap (C_i \cup \{v_i\}) \neq \emptyset$ for all $i \in
  \{1, \ldots, k\}$. By setting $\sigma_V = \emptyset$, we especially get that
  there exists an interpretation $\sigma_X \subseteq \varsX$ such that for any
  interpretation $\sigma_Y \subseteq \varsY$ there exists an interpretation
  $\sigma_Z \subseteq \varsZ$ such that $(\widehat{\sigma_X} \cup
  \widehat{\sigma_Y} \cup \{\neg v_1, \ldots, \neg v_k\} \cup
  \widehat{\sigma_Z}) \cap (C_i \cup \{v_i\}) \neq \emptyset$ for all $i \in
  \{1, \ldots, k\}$. Since the only literals containing a $v_i$ variable on the
  left-hand side of the $\cap$ are negative and the only ones on the right-hand
  side are positive, we get $(\widehat{\sigma_X} \cup \widehat{\sigma_Y} \cup
  \widehat{\sigma_Z}) \cap C_i \neq \emptyset$ for all $i \in \{1, \ldots, k\}$,
  which establishes validity of $\Theta$.
\end{proof}

Now it is straightforward to generalize the reduction from \cite{ai:ShenE16}:
let $\Theta$ be a QBF and apply the reduction from \cite{ai:ShenE16} to the
restricted QBF $\Theta^\uparrow$.

\begin{theorem}
  Let $\Theta$ be a QBF. Let the ELP $\Pi_\Theta$ be obtained by applying the
  reduction by Shen and Eiter \shortcite{ai:ShenE16} to the restricted QBF $\Theta^\uparrow$. It
  holds that $\Theta$ is valid iff $\Pi_\Theta$ is consistent, that is,
  $\Pi_\Theta$ has at least one candidate world view.
\end{theorem}

Correctness of this theorem follows from immediately
from 
Proposition~\ref{prop:sheneiter} and
Proposition~\ref{prop:validityequivalence}. 
\newcommand{\selp}{\emph{selp}}
\section{The \selp{} System}\label{sec:system}

\newcommand{\easptoasp}{\mbox{\emph{easp2asp.py}}}
\newcommand{\groupWorldViews}{\mbox{\emph{groupWorldViews.py}}}
\newcommand{\easpGrounder}{\mbox{\emph{easpGrounder.py}}}

We implemented the reduction in Section~\ref{sec:encoding} as part of the single shot
ELP solving toolbox \selp, available at \linkselp. In addition, the toolbox
features a grounder for ELPs and a grouping script which groups answer sets
of the reduction into candidate world views (allowing for enumeration).
The tools are implemented in \emph{python} and depend on the parser generator
\emph{LARK}\footnote{Available here: \url{https://github.com/erezsh/lark}}, 
the rule decomposition tool \emph{lpopt} \cite{lopstr:BichlerMW16},
the tree decomposition tool \emph{htd\_main} \cite{cpaior:AbseherMW17}, 
and the plain ASP grounder \emph{gringo} \cite{lpnmr:GebserKKS11}.

\paragraph{Input Formats.} The \selp{} solver reads the
\emph{EASP-not} file format, which is a restriction of the ASP input language of
\emph{gringo} to plain
ground logic programs as defined in Section~\ref{sec:preliminaries}, extended
with the \texttt{\$not\$} operator for epistemic
negation. This allows us to encode ELPs as defined in
Section~\ref{sec:preliminaries}.
\selp{} also supports
\emph{EASP-KM}, defined by adding the operators \texttt{K\$} and \texttt{M\$} instead of \texttt{\$not\$}.
By allowing variables in body elements, both formats also have a
\emph{non-ground} version.
The toolbox offers
scripts to translate between 
the two formats.

\paragraph{Toolbox.}
We briefly present the main building blocks of \selp{}.

\smallskip\textbf{\easpGrounder}
takes as input a \emph{non-ground EASP-not} program and outputs its equivalent
ground form by rewriting it into an ASP program that the \emph{gringo} grounder can
understand and ground. This is done by encoding epistemic negation as predicate
names and, after grounding, re-introducing epistemic negation where a
placeholder predicate appears.
Our grounding component, \easpGrounder{}, supports arithmetics and the
\emph{sorts} format \cite{logcom:KahlWBGZ15} as input.

\smallskip\textbf{\easptoasp} is \selp's key component. It takes a ground
\emph{EASP-not} program, performs the reduction given in
Section~\ref{sec:reduction} (with some modifications to account for the extended
language of ASP used by today's ASP systems and some straightforward
optimizations), also adhering to the practical considerations presented in
Section~\ref{sec:reduction-practice}, and it finally outputs the resulting
non-ground logic program in the syntax of \emph{gringo}. Optionally, additional
\emph{clasp} directives are generated to allow for enumeration; cf.\ Section
\ref{sec:relatedwork}. For concrete implementation details, please consult the
freely available source code at \linkselp.

\smallskip\textbf{\groupWorldViews}
takes \emph{clasp}'s output in JSON format, groups the answer sets into
candidate world views according to their $\relation{g}(\cdot)$ atoms, and
outputs them in a human-readable format.

\paragraph{Usage.} As a typical use case, suppose the file
\texttt{problem.easp} contains a non-ground ELP encoding of a problem of
interest and the file \texttt{instance.easp} contains a problem instance.
In order to output all candidate world views, one would use
the following command
(flags \verb!-pas! and \verb!--project! enable projection of answer sets onto
relevant predicates only. \verb!-n0! tells \emph{clasp} to compute all answer
sets, and \verb!--outf=2! to print in JSON format. 
\emph{lpopt} is used to decompose long
rule bodies. The \verb!--sat-prepro=3! flag is recommended by \emph{lpopt}):

\bigskip\noindent
\verb!cat problem.easp instance.easp |!\\
\verb!easpGrounder.py -sELP | easp2asp.py -pas |!\\
\verb!lpopt | gringo | clasp -n0 --outf=2 --project --sat-prepro=3 |!\\
\verb!groupWorldViews.py!
\section{Experimental Evaluation}\label{sec:experiments}

We tested our system \selp{} against the state-of-the-art ELP solver,
\emph{EP-ASP}~\cite{ijcai:SonLKL17}, using three test sets.  For every test set,
we measured the time it took to solve the consistency problem. For \selp{},
\emph{clasp} was stopped after finding the first answer set. For \emph{EP-ASP},
search was terminated after finding the first candidate world view\footnote{Note
that to have a fair comparison we disabled the subset-maximality
check on the guess that \emph{EP-ASP} performs by default.}. Note that a single
answer set of the \selp{} system is enough to establish consistency of an input
ELP. \emph{EP-ASP} needs to compute a full candidate world view to be able to
prove consistency.

Experiments were run on a 2.1GHz AMD Opteron 6272 system with 224
GB of memory. Each process was assigned a
maximum of 14 GB of RAM. For \emph{EP-ASP}, we used the required \emph{clingo}
4.5.3, since newer versions are incompatible with the solver. For \selp{}, we used \emph{clingo} 5.2.2, \emph{htd\_main} 1.2.0, and \emph{lpopt} 2.2.  
The time it took \emph{EP-ASP} to rewrite the input to its own internal
format was not measured. \emph{EP-ASP} was called with the preprocessing option
for brave and cautious consequences on, since it always ran faster 
this way. The \selp{} time is the sum of running times of its components.

\begin{figure}[t]
\begin{subfigure}{0.45\linewidth}
    \begin{tikzpicture}[scale=0.5]
    \begin{semilogyaxis}[
        width=2\linewidth,
        grid=none,
        xmin=1,
        xmax=25,
        ymin=0.01,
        ymax=28800,
        xlabel={number of students},
        ylabel={time [sec]},
        legend pos=north west
        ]
        \addplot[scatter/classes={
            selp={mark=x,blue},
            EP-ASP={mark=o,red}
            },scatter,only marks,scatter src=explicit symbolic]
        table[x=num_students,y=time,meta=solver]
            {plotdata/eligibleK.dat}; 
        \legend{$\selp{}$,\emph{EP-ASP}}
    \end{semilogyaxis}
    \end{tikzpicture}
    \caption{Scholarship Eligibility}
    \label{fig:eligibility}
\end{subfigure}
\hspace*{\fill} %
\begin{subfigure}{0.45\linewidth}
    \begin{tikzpicture}[scale=0.5]
    \begin{semilogyaxis}[
        width=2\linewidth,
        grid=none,
        xmin=1,
        xmax=25,
        ymin=0.1,
        ymax=1800,
        xlabel={plan length},
        ylabel={time [sec]},
        legend pos=north east
        ]
        \addplot[scatter/classes={
            selp={mark=x,blue},
            EP-ASP={mark=o,red}
            },scatter,only marks,scatter src=explicit symbolic]
        table[x=plan_length,y=time,meta=solver]
            {plotdata/yale.dat}; 
        \legend{$\selp{}$,\emph{EP-ASP}}
    \end{semilogyaxis}
    \end{tikzpicture}
    \caption{Yale Shooting}
    \label{fig:yale}
\end{subfigure}
\vspace{-0.16cm}
\caption{Benchmark results. Missing points indicate timeouts.}
\label{fig:benchmarks}
\vspace{-0.14cm}
\end{figure}
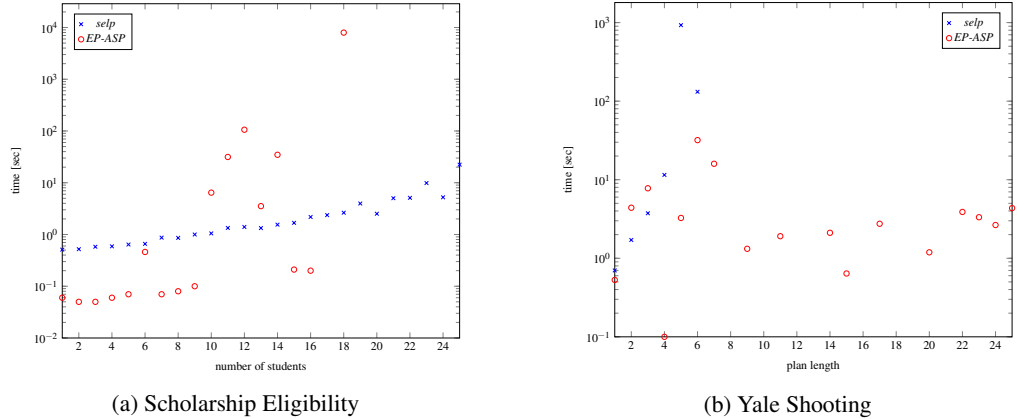

\paragraph{Benchmark Instances.}
We used three types of benchmarks, two coming from the ELP literature and one from the QSAT domain that 
contains structures of low treewidth\footnote{Benchmark archive: \linkbench}.
\subparagraph{Scholarship Eligibility (SE).} 
This set of non-ground ELP programs is shipped together with \emph{EP-ASP}.
Its instances encode the scholarship eligibility problem %
for 1 to 25
students. 

\subparagraph{Yale Shooting (YS).} 
This test set consists of 25 non-ground ELP programs 
encoding a simple version of the Yale Shooting Problem, a conformant planning
problem: the only uncertainty is whether the gun is initially loaded or not, and
the only fluents are the gun's load state and whether the turkey is alive.
Instances differ in the time horizon. We follow the ELP encoding from 
\cite{logcom:KahlWBGZ15}. 

\subparagraph{Tree QBFs (TQ).} 
The hardness proof for ELP consistency \cite{ai:ShenE16}
relies on a reduction from the validity problem for 
restricted 
quantified boolean formulas 
with three quantifier blocks (i.e.\ 3-QBFs), which can 
be generalized to arbitrary 
 3-QBFs as discussed in Section~\ref{sec:qbf}.
We apply this extended reduction to the
14
``Tree'' instances
of QBFEVAL'16
\cite{Pulina16},
available at
\url{http://
www.qbflib.org/family_detail.php?idFamily=56}, 
splitting each instance's variables into
three random quantifier blocks.

\paragraph{Results.}
The results for the first two sets are shown in Figure~\ref{fig:benchmarks}.
\selp{} solves all instances 
from (SE) %
within 30 seconds, while \emph{EP-ASP} only solves 17 within the time limit of 8 hours.
For (YS),
on the other hand,
\selp{} is able to solve only 6 instances within the time limit of 30 minutes, whereas \emph{EP-ASP} can solve 17.
Finally, for %
(TQ),
\selp{} can solve 6 of the 14 instances
within the time limit of 12 hours, 
whereas \emph{EP-ASP} was unable to solve any instances at all.
These results confirm that \selp{} is highly competitive on well-structured problems: in 
the (SE)
instances, knowledge about students is not interrelated, and hence the graph
$G_\Pi$ of the ground ELP $\Pi$ consists of one component for each student, thus
having constant treewidth. The 
(TQ) instances keep their constant treewidth thanks to the fact that both
the reduction from QBF to ELP 
(cf.\ Section~\ref{sec:qbf})
and from ELP to non-ground ASP (cf.\
Section~\ref{sec:reduction-practice}) preserve the low treewidth of the
original QBF instance.
Different from \selp{}, \emph{EP-ASP} is not designed to exploit such structural
information of ELPs and, consequently, performs worse than \selp{} in these
benchmarks.
On the other hand, (YS) contains instances of high treewidth, even though it
does not depend on the horizon. \emph{EP-ASP} is therefore able to outperform
\selp{} on such instances. A similar observation can be made for the ``Bomb in
the Toilet'' problem, as benchmarked in \cite{ijcai:SonLKL17}, which inherently
contains a huge clique structure.  \selp{} is not designed to solve such
instances, and is therefore most suited to solve ELPs of low treewidth, where it
is able to efficiently exploit the problem structure.
\section{Conclusions}\label{sec:conclusions}

In this paper, we have seen that ELPs can be encoded into ASP programs using
long non-ground rules, such that a single call to an ASP solver is sufficient to
evaluate them. A prototype ELP solver implementation, \selp, performs
particularly well on problems whose internal structure is of low treewidth. A
combined solver that either calls \selp{} or another state-of-the-art solver
based on the treewidth of the input may therefore lead to even better overall
performance.

Another topic for future work is that, under the FLP semantics, checking whether
a given atom $a$ is true in all candidate world views with a
\emph{subset-maximal} guess $\Phi$ is known to be \SIGMA{P}{4}-complete
\cite{ai:ShenE16}. To solve this problem, advanced optimization features of
state-of-the-art ASP solvers could allow us to encode this subset-maximality
condition, while leaving the core of our encoding unchanged.

Finally, an interesting question is program optimization. Recently, a practical,
easily applicable notion of strong equivalence for ELPs has been defined
\cite{aaai:FaberMW19}. It would be interesting to investigate if and how parts
of ELPs can be replaced in such a way that the solving performance of \selp{}
improves, seeing that \selp{} is sensitive to treewidth. This could lead to an
encoding technique for ELPs that tries to minimize the treewidth, similar to the
class of connection-guarded ASP \cite{ijcai:BliemMMW17}, which was recently
proposed in order to write ASP programs in such a way as to keep the treewidth
of the resulting ground program low.
 
\section*{Acknowledgements}
This work was funded by the Austrian Science Fund (FWF) under
grant numbers Y698 and P30930.


\bibliographystyle{acmtrans}
\bibliography{references}

\end{document}